\documentclass[pdflatex,sn-aps]{sn-jnl}

\usepackage{graphicx}%
\usepackage{multirow}%
\usepackage{amsmath,amssymb,amsfonts}%
\usepackage{amsthm}%
\usepackage{mathrsfs}%
\usepackage[title]{appendix}%
\usepackage{textcomp}%
\usepackage{manyfoot}%
\usepackage{booktabs}%
\usepackage{algorithm}%
\usepackage{algorithmicx}%
\usepackage{algpseudocode}%
\usepackage{listings}%

\usepackage{xcolor,mathtools}
\def\BibTeX{{\rm B\kern-.05em{\sc i\kern-.025em b}\kern-.08em
    T\kern-.1667em\lower.7ex\hbox{E}\kern-.125emX}}

\usepackage{tikz}
\usetikzlibrary{plotmarks}
\usepackage{pgfplots}
\usetikzlibrary{calc}
\usetikzlibrary{shapes,arrows}
\usetikzlibrary{decorations.markings}
\usetikzlibrary{positioning}
\pgfplotsset{compat=1.10}
\usetikzlibrary{calc}
\usetikzlibrary{shapes,arrows}
\usetikzlibrary{decorations.markings}
    	
\usepackage{color, colortbl}
\usepackage{caption}
\DeclareMathOperator{\erf}{erf} 
\DeclareMathOperator{\maxtr}{maxtr}
\DeclareMathOperator*{\argmax}{arg\,max}

\theoremstyle{thmstyleone}%
\newtheorem{theorem}{Theorem}

\theoremstyle{thmstyletwo}%
\newtheorem{remark}{Remark}%

\theoremstyle{thmstylethree}%
\newtheorem{definition}{Definition}%

\begin{document}

\title[Randomized Distributed Function Computation (RDFC): Ultra-Efficient Semantic Communication Applications to Privacy]{Randomized Distributed Function Computation (RDFC): Ultra-Efficient Semantic Communication Applications to Privacy}

\author[1,2]{\fnm{Onur} \sur{G{\"u}nl{\"u}}}\email{onur.guenlue@tu-dortmund.de}

\affil[1]{\orgdiv{Lehrstuhl f{\"u}r Nachrichtentechnik}, \orgname{Technische Universit{\"a}t Dortmund}, \country{Germany}}
\affil[2]{\orgdiv{Information Coding Division}, \orgname{Link{\"o}ping University}, \country{Sweden}}

\abstract{
We establish the randomized distributed function computation (RDFC) framework, in which a sender transmits just enough information for a receiver to generate a randomized function of the input data. Describing RDFC as a form of semantic communication, which can be essentially seen as a generalized remote‑source‑coding problem, we show that security and privacy constraints naturally fit this model, as they generally require a randomization step. Using strong coordination metrics, we ensure (local differential) privacy for every input sequence and prove that such guarantees can be met even when no common randomness is shared between the transmitter and receiver.

This work provides lower bounds on Wyner's common information (WCI), which is the communication cost when common randomness is absent, and proposes numerical techniques to evaluate the other corner point of the RDFC rate region for continuous‑alphabet random variables with unlimited shared randomness. Experiments illustrate that a sufficient amount of common randomness can reduce the semantic communication rate by up to two orders of magnitude compared to the WCI point, while RDFC without any shared randomness still outperforms lossless transmission by a large margin. A finite blocklength analysis further confirms that the privacy parameter gap between the asymptotic and non-asymptotic RDFC methods closes exponentially fast with input length. Our results position RDFC as an energy-efficient semantic communication strategy for privacy‑aware distributed computation systems.
}


\maketitle

\section{Introduction}\label{sec1}
The need to transmit a message's intended meaning rather than its raw bits has led to the framework of semantic communication \cite{PetarSemantic,GunduzSemantic}. In this framework, a reconstructed signal is evaluated by a measure of quality tied to its semantics \cite{GunduzBITS}. Unlike conventional physical‑layer schemes that are oblivious to content, semantic communication transmits only the information that is relevant. Transmitting objects and their relative positions in an image, instead of the full pixel map, is a semantic communication example \cite{GunduzBITS}. Since the semantic description at the transmitter can be regarded as a function of the data, semantic communication can naturally be modeled as a remote (hidden) source‑coding problem, where what is transmitted is not the data itself but a function of it that is unknown to the transmitter \cite{DobrushinRemote,BergerBook,Csiszarbook}. This formulation generalizes the lossy source coding problem, allowing the fidelity measure to be tailored to semantic relevance.

We address a broader remote‑source‑coding scenario in which a randomized function of the data is used for distributed function computation. The receiver, thus, evaluates a function whose inputs are randomized versions of the transmitters' observations, establishing the \emph{randomized distributed function computation (RDFC) framework with semantic communications}. As security and privacy methods generally require randomization, such constraints are natural use-cases for RDFC \cite{OurSecTutorial,InterativeSecureComp,GustafISIT2025}. Specifically, we study privacy‑constrained RDFC problems, where the transmitter seeks to let the receiver remotely simulate a sequence so that the input and output sequences follow a {given} target joint distribution that preserves privacy. Minimizing the communication load required for such remote simulation is known in the literature as a \emph{coordination problem}, also called distributed channel synthesis or simulation \cite{CuffChannelSynthesis,OnurEUCNC2025,ReverseShannon,HarshaOriginalOneShot,GerhardChannelSimulation}. Unlike recent work that often targets \emph{empirical (weak) coordination} for machine‑learning and neural compression tasks \cite{HavasiMRC,PeterLDP,Khisti,GunduzAdaptive}, we focus on \emph{strong coordination} measures \cite{CuffChannelSynthesis} to design RDFC methods. Thus, RDFC provides coordination guarantees for each computation instance, rather than for the average behavior, by imposing a joint typicality constraint, crucial for robust security and privacy to not enable new attacks.

{
\subsection{Summary of Contributions}
Unlike deterministic function computation tasks, \emph{randomized} function computation problems experience a significant reduction in the required communication load when transmitters and receivers share common randomness \cite{MACChannelSimulation}. In a point-to-point setting, the coordination-randomness rate region has two corner points. Without common randomness, the minimum communication rate corresponds to Wyner's common information (WCI) $C(\widetilde{X};Y)=\inf_{U: \widetilde{X}-U-Y} I(\widetilde{X},Y;U)$ between the channel input $\widetilde{X}$ and output $Y$ \cite{WCI}. With enough common randomness, the minimum communication rate reduces to the mutual information $I(\widetilde X;Y)$. Building on this observation, the main contributions of this work include the following:

\begin{itemize}
  \item We introduce and formalize the RDFC framework as a semantic communication model for privacy-constrained distributed computation, which is an extension of the remote source coding problem. Within a strong coordination setting, we interpret privacy mechanisms as randomized functions of given input sequences, unlike in existing randomized privacy mechanisms, and identify the two corner operating points of the coordination-randomness rate region for such privacy problems.

  \item We analyze a continuous-alphabet RDFC scenario in which the randomized output must satisfy a local differential-privacy (LDP) constraint. For a clipped Gaussian input and a Gaussian $(\epsilon,\delta)$-LDP mechanism, we impose symmetry to derive a lower bound on the WCI and develop a numerical procedure to evaluate the mutual information $I(\widetilde X;Y)$. The resulting comparisons quantify the communication-rate savings achievable due to common randomness under a LDP constraint and reveal parameter regimes where the new WCI lower bound improves upon generic lower bounds. We also consider parameter sets where one can identify the individual effect of the the LDP parameter $\epsilon$.

  \item We study a discrete RDFC scenario based on a symmetric random response mechanism that effectively combines several randomized bit responses. For a family of binary symmetric channel (BSC)-mixture models, we derive a lower bound on the WCI using symmetry and Witsenhausen-type techniques and compare it with the mutual information $I(\widetilde X;Y)$, which itself is a lower bound on the WCI. The resulting examples demonstrate that common randomness can substantially reduce the semantic communication rate, while RDFC without any shared randomness still significantly outperforms lossless transmission in terms of communication load. These communication rate gains indicate substantial improvements in energy efficiency.
\end{itemize}

A subset of these results were reported in the conference version of this work in \cite{MyWIFS2024}. Moreover, in this work, we introduce the following contribution:
\begin{itemize}
  \item We establish a finite-blocklength analysis for RDFC under LDP constraints. Specifically, for any rate $R > I(\widetilde X;Y)$ in the interior of the coordination-randomness rate region, we provide explicit bounds on the RDFC performance of random encoder-decoder pairs. This fundamental result, which is in the form of error-exponent analysis, yields achievable LDP parameters $(\epsilon,\delta_n)$ with $\delta_n \to \delta$ at an exponential rate in the blocklength $n$, showing that non-asymptotic RDFC encoder-decoder pairs with sufficient common randomness can approach the asymptotic privacy guarantees with exponentially vanishing loss.
\end{itemize}

These results illustrate that using RDFC encoder-decoder pairs, instead of first introducing randomness and then compressing the randomized output, one has a principled way to design privacy-aware distributed computation systems with significantly smaller communication load, leading to substantial energy savings.

}

\subsection{Paper Organization}
The paper is organized as follows. Section~\ref{sec:preliminaries} introduces the RDFC framework that yields per‑instance function computation guarantees and the LDP metric {for discrete random variables}. Section~\ref{sec:LDP} establishes a lower bound on the WCI and a numerical procedure for evaluating the mutual information $I(\widetilde{X};Y)$ when continuous‑valued random variables are subject to an LDP constraint. Section~\ref{sec:nonasympLDP} establishes achievable LDP parameters for finite‑length RDFC encoder–decoder pairs for given target LDP parameters. Section~\ref{sec:SymmetricforRR} considers a random response mechanism and derives a lower bound on the WCI for a family of discrete symmetric channels. Section~\ref{sec:conclusions} concludes the paper.

{
\subsection{Notation}
Random variables are denoted by uppercase letters $X$ and their realizations by lowercase letters $x$. A random variable $X$ follows either a probability mass function (pmf) $P_X$ with support $\operatorname{supp}(P_X)$ or has a probability density function (pdf) $p_X$. Sets are represented using calligraphic letters $\mathcal{X}$ and their cardinality is written as $|\mathcal{X}|$. Sequences of length $n$ are expressed as $X^n = (X_1, X_2, \ldots,X_i,\ldots,  X_n)$. 

The total variation (TV) distance between two distributions $P_Y$ and $P_X$ is defined as
\begin{align}
    \left\|P_Y-P_X\right\|_{\text{TV}} \triangleq \frac{1}{2}\sum_{b\in \mathcal{B}} |P_Y(b) - P_X(b)|.
\end{align}

Define the information density 
\begin{align}
    \imath_{\widetilde{X};Y}(\widetilde{x},y)= 
\log \frac{\mathrm d P_{\widetilde{X}Y}}{\mathrm d P_{\widetilde{X}} P_Y}
\end{align}
where $\frac{\mathrm d P_{\widetilde{X}Y}}{\mathrm d P_{\widetilde{X}} P_Y}$ is the Radon-Nikodym derivative of $P_{\widetilde{X}Y}$ with respect to $P_{\widetilde{X}} P_Y$. Similarly, define the $\alpha$-mutual information, for any $\alpha>1$,
\begin{align}
    I_{\alpha}(\widetilde{X};Y)
  =\frac{\alpha}{\alpha-1}\,
   \log\,
   \mathbb{E}\!\Bigl[
       \mathbb{E}^{\frac{1}{\alpha}}\!\Bigl[
           e^{\alpha\imath_{\widetilde{X};Y}(\widetilde{X},\widetilde{Y})}
           \Bigm|
           \widetilde{Y}
       \Bigr]
   \Bigr]
\end{align}
where $(\widetilde{X},\widetilde{Y})\sim P_{\widetilde{X}} P_Y$. 
}

\section{Preliminaries}\label{sec:preliminaries}
\subsection{Randomized Distributed Function Computation (RDFC): Semantic Communications via Strong Coordination}
The RDFC framework, leveraging strong coordination methods, seeks to coordinate the sequences of multiple nodes through channels while transmitting only the minimum necessary amount of information. Because coordination is established by exhibiting, in those local sequences, a joint behavior summarized by {any given} target joint probability distribution that depends on each node's goals as well as computation and communication resources, reliable communication of the raw local sequences to other nodes is unnecessary. We remark that RDFC methods can also be implemented by using low‑complexity autoencoders as encoder-decoder pairs \cite{DidrikISIT2025}.

\begin{figure}[t!]
\centering
\ttfamily
\resizebox{0.75\linewidth}{!}{
		\begin{tikzpicture}
		\node (input) at (0,0) {$\widetilde{X}^n$};
		\node (cr) at (2.5,-3) {$C\in\{1,2,\ldots,2^{nR_0}\}$};
		\node (enc) at (0,-1.5) [draw,fill=blue!15!white, rounded corners = 3pt, minimum width=1.2cm,minimum height=0.6cm, align=left] {Enc};
		\node (dec) at (5,-1.5) [draw,fill=blue!15!white,rounded corners = 3pt, minimum width=1.2cm,minimum height=0.6cm, align=left] {Dec};
		\node (output) at (5,0) {$Y^n$};
		\node (const) at (7.2,-0.048) {s.t. $(\widetilde{X}^n,Y^n)\sim Q^n_{\widetilde{X}Y}$};
		\draw[decoration={markings,mark=at position 1 with {\arrow[scale=1.5]{latex}}},
		postaction={decorate}, thick, shorten >=1.4pt] (input.south) -- (enc.north);
		\draw[decoration={markings,mark=at position 1 with {\arrow[scale=1.5]{latex}}},
		postaction={decorate}, thick, shorten >=1.4pt] (enc.east) -- (dec.west) node [above, midway] {$S\in \{1,2,\ldots,2^{nR}\}$};
		\draw[decoration={markings,mark=at position 1 with {\arrow[scale=1.5]{latex}}},
		postaction={decorate}, thick, shorten >=1.4pt] (dec.north) -- (output.south);
		\draw[decoration={markings,mark=at position 0.9999 with {\arrow[scale=1.5]{latex}}},
		postaction={decorate}, thick, shorten >=0.4pt] (cr.west) to [out=180,in=270] (enc.south);
		\draw[decoration={markings,mark=at position 0.9999 with {\arrow[scale=1.5]{latex}}},
		postaction={decorate}, thick, shorten >=0.4pt] (cr.east) to [out=0,in=270] (dec.south);
		\end{tikzpicture}
	} \caption{A two–node RDFC model where both nodes may share common randomness $C$. Within the strong coordination framework, the receiver uses the transmitted index $S$ and $C$ to output a sequence $Y^{n}$. This procedure ensures that every pair $(\widetilde{x}^{n},y^{n})$ follows {any} target joint distribution $Q^{n}_{\widetilde{X}Y}$ given in Eq.~(\ref{eq:coordinationconstraint}) below, enabling a cooperation strategy that benefits both nodes. The design goal focuses on minimizing the channel rate $R$ to reduce function computation latency and energy consumption.}
    \label{fig:strongcoordsimpleproblem} 
        \vspace*{-0.1cm}
\end{figure}
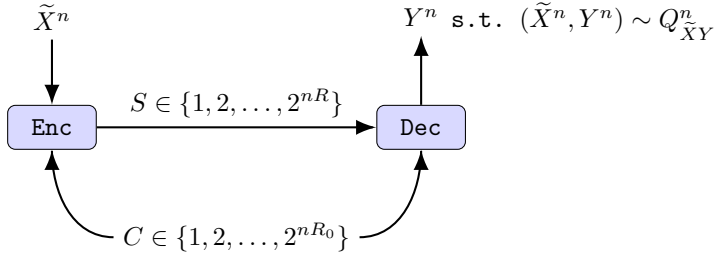

Consider the two‑user RDFC model illustrated in Fig.~\ref{fig:strongcoordsimpleproblem}. After observing a sequence $\widetilde{X}^n=\widetilde{x}^n$, the transmitter jointly encodes $\widetilde{x}^n$ together with common randomness $C\in\{1,2,\ldots,2^{nR_0}\}$, shared between the transmitter and receiver and obtainable by using physical unclonable functions \cite{bizimGlobalSIP}. This encoding produces an index $S\in\{1,2,\ldots,2^{nR}\}$; the transmitter sends $S$ over a noiseless link to the receiver, which uses $S$ and $C$ to synthesize a sequence $Y^n=y^n$ such that $(\widetilde{x}^n,y^n)\sim Q^n_{\widetilde{X}Y}$. We next define the coordination–randomness region for the two‑user RDFC model shown in Fig.~\ref{fig:strongcoordsimpleproblem}{, where we consider discrete random variables.}

\begin{definition}\label{def:regiondef}
    A rate pair $(R,R_{0})$ is \emph{achievable} for $Q_{\widetilde{X}Y}$ if, for any $\epsilon_{n}>0$, there exist a blocklength $n\ge 1$ and a matching encoder–decoder pair such that
    \begin{align}
        \lim_{n \to \infty} \left\| P_{\widetilde{X}^nY^n}{({\widetilde{x}^n,y^n})} - \prod_{i=1}^n Q_{\widetilde{X}Y}{({\widetilde{x}_i,y_i})} \right\|_{\text{TV}} = 0\label{eq:coordinationconstraint}
    \end{align}
    where $P_{\widetilde{X}^{n}Y^{n}}$ denotes the distribution induced by that pair.

    The \emph{coordination–randomness region} $\mathcal{R}_{\text{RDFC}}$ is the closure of all achievable $(R,R_{0})$ tuples for the two‑user RDFC model given in Fig.~\ref{fig:strongcoordsimpleproblem}. \hfill $\lozenge$
\end{definition}

We next provide the coordination-randomness region $\mathcal{R}_{\text{RDFC}}$.

\begin{theorem}[{\hspace{1sp}\cite[Theorem~II.1]{CuffChannelSynthesis}}]\label{theo:CoordinationRegion}
    The coordination-randomness region $\mathcal{R}_{\text{RDFC}}$ is the union over all $P_{\widetilde{X}UY} = P_UP_{\widetilde{X}|U}P_{Y|U}$ of the rate pairs $(R,R_0)$ satisfying 
    \begin{align}
        & R\geq I(\widetilde{X};U),\label{eq:Rbound}\\
        & R+ R_0 \geq I(\widetilde{X},Y;U),\label{eq:RR0bound}\\
        & \sum_{u\in\mathcal{U}}P_{\widetilde{X}UY} = Q_{\widetilde{X}Y}.\label{eq:preserve_QXtildeY}
    \end{align}
    
\end{theorem}

{
\begin{remark}
    The random variable $U$ in (\ref{eq:Rbound}) and (\ref{eq:RR0bound}) is any auxiliary random variable that satisfies $P_{\widetilde{X}UY} = P_UP_{\widetilde{X}|U}P_{Y|U}$ and (\ref{eq:preserve_QXtildeY}). For each such auxiliary random variable $U$, one can evaluate the lower bounds (\ref{eq:Rbound}) and (\ref{eq:RR0bound}) to obtain a valid operating tuple $(R,R_0)$. Thus, the union of all achievable operating tuples, obtained by considering all valid auxiliary random variables, establishes the coordination-randomness region $\mathcal{R}_{\text{RDFC}}$.
\end{remark}

}

Theorem~\ref{theo:CoordinationRegion} remains valid to describe the coordination–randomness region for the continuous‑valued random variables analyzed below. A standard discretization argument \cite[Remark~3.8]{Elgamalbook} carries the achievability proof over to well‑behaved continuous alphabets, including the (truncated) Gaussian distributions considered in this work; see also \cite[Sec.~VII]{CuffChannelSynthesis}.

Within the coordination-randomness region, Theorem~\ref{theo:CoordinationRegion} identifies two corner points.
 
First, consider the case $R_{0}=0$, meaning the transmitter and receiver do not share any common randomness. Combining \eqref{eq:Rbound} and \eqref{eq:RR0bound} yields the first corner point
\begin{align}
    R\geq C(\widetilde{X};Y)=\inf_{U:\; \widetilde{X}-U-Y} I(\widetilde{X},Y;U)\label{eq:WCIbound}
\end{align}
where $C(\widetilde{X};Y)$ denotes the WCI between $\widetilde{X}$ and $Y$ \cite{WCI}. Note that in many multi‑user information theory results, one optimizes over auxiliary random variables $V$ that satisfy Markov chains like $V-\widetilde{X}-Y$ and several techniques exist for that task. Here, however, the auxiliary random variable $U$ forms the Markov chain $\widetilde{X}-U-Y$, making the optimization of $P_{\widetilde{X}UY}$ considerably harder \cite{WCIMethod}. We derive lower bounds on $C(\widetilde{X};Y)$ for symmetric distributions of practical relevance.

Second, consider the opposite extreme, where the two terminals have sufficient common randomness $C$. Under this assumption, \eqref{eq:Rbound} and \eqref{eq:RR0bound} yield the second corner point
\begin{align}
    R\geq I(\widetilde{X};Y)\label{eq:MIbound}
\end{align}
since the data‑processing inequality implies $I(\widetilde{X};U)\geq I(\widetilde{X};Y)$. One can achieve the corner point in \eqref{eq:MIbound} with $R_{0}\ge H(Y|X)$, although \cite{CuffCoordinationCapPaper} shows that some scenarios need even less common randomness. The bound in \eqref{eq:MIbound} recovers the reverse Shannon theorem \cite{ReverseShannon}, which shows that an unlimited amount of common randomness and a noiseless channel suffice to simulate a noisy channel. This is the reverse of the operation in the channel coding theorem in Shannon's seminal work \cite{Shannon1948}.

From \cite{WCI}, we have the chain of inequalities  
\begin{align}
    I(\widetilde{X};Y)\leq C(\widetilde{X};Y)\leq \min\{H(\widetilde{X}),\; H(Y)\}\label{eq:MIlessthanWCI}
\end{align}
which underscores the significant communication rate savings of RDFC. Even in the absence of common randomness, an RDFC method can surpass both standard baselines: (i) Transmitting the sequence $\widetilde{X}^{n}$ losslessly at rate $H(\widetilde{X})$ and letting the receiver compute the randomized function output $Y^{n}$, and (ii) Computing $Y^{n}$ at the transmitter and then transmitting it losslessly at rate $H(Y)$.

\subsection{Local Differential Privacy}
LDP, which we consider to guarantee individual user privacy, naturally forms an RDFC problem, because the privacy guarantee arises from computing a randomized function of the input.

Next, we define $(\epsilon,\delta)$--LDP, often called \emph{approximate} LDP as $\delta\neq 0$. Let $\mathcal{M}(\cdot)$ be a randomized mechanism that assigns every $\widetilde{x}\in\mathcal{\widetilde{X}}$ to a probability distribution $\mathcal{M}(\cdot|\widetilde{x}) \in \mathcal{P}(\mathcal{Y})$.

\begin{definition}[\hspace{1sp}\cite{FlavioLDP,evfimievski2003limiting, 4690986}]
    A mechanism $\mathcal{M}: \mathcal{\widetilde{X}} \rightarrow \mathcal{P}(\mathcal{Y})$ is $(\epsilon,\delta)$-LDP for $\epsilon \geq 0$ and $\delta \in [0,1]$ if we have
    \begin{align}
        \sup_{\widetilde{x},\widetilde{x}' \in \mathcal{\widetilde{X}}} \sup_{Y \subseteq \mathcal{Y}} \big( \mathcal{M}(Y|\widetilde{x}) - e^\epsilon \mathcal{M}(Y|\widetilde{x}') \big) \leq \delta.
    \end{align} \hfill $\lozenge$
\end{definition}

In Section~\ref{sec:LDP}, we formulate an RDFC problem that synthesizes a Gaussian LDP mechanism $\mathcal{M}(\cdot)$, where by adding independent Gaussian noise, $\mathcal{M}(\cdot)$ satisfies a desired $(\epsilon,\delta)$--LDP guarantee. Due to its low complexity, the Gaussian mechanism often acts as the privacy layer in empirical risk minimization methods that run optimization routines such as stochastic gradient descent or the Adam optimizer \cite{Bassily}.

\section{RDFC for Local Differential Privacy with Additive Gaussian Noise}\label{sec:LDP}
{Next, we consider continuous-valued random variables.} Suppose an i.i.d. Gaussian random variable $X^n\sim \mathcal{N}^n(0,\sigma_x^2)$ clipped to the range $[-C,C]$, where $C>0$. Define, for $a\in R$,
\begin{alignat}{2}
    & \beta = \frac{C}{\sigma_X},\\ 
    &\erf(a) = \frac{2}{\sqrt{\pi}}\int_{0}^a e^{-t^2}dt,\\
    & \phi(a) = \frac{1}{\sqrt{2\pi}}e^{-a^2/2},\\
    &\gamma(a) = \frac{a \phi (a)}{\erf\big(\frac{a}{\sqrt{2}}\big)}.
\end{alignat}
The clipped output $\widetilde{X}^n$ is an i.i.d. truncated Gaussian random variable with zero mean and variance
\begin{align}
    \sigma_{\widetilde{X}}^2 = \sigma_X^2 \Big(1-2\gamma(\beta)\Big).\label{eq:sigmaXsquarevssigmaX}
\end{align}
The transmitter observes the clipped sequence $\widetilde{X}^n$ as its input and aims to achieve $(\epsilon,\delta)$-LDP by considering the Gaussian LDP mechanism {\cite{DworkDP}, such that the mechanism output $Y^n$ observed at the receiver is equal to $\widetilde{X}^n + \widetilde{Z}^n$, where $\widetilde{X}^n$ and $\widetilde{Z}^n$ are independent and we have} 
\begin{align}
    \{(Y^n-\widetilde{X}^n)|\widetilde{X}^n\}\vcentcolon= \{\widetilde{Z}^n|\widetilde{X}^n\}= \widetilde{Z}^n\sim \mathcal{N}^n(0,\sigma_{\widetilde{Z}}^2).
\end{align}
If we have
\begin{align}
    \sigma_{\widetilde{Z}}^2 = \frac{8C^2}{\epsilon^2}\log\big(\frac{1.25}{\delta}\big),
\end{align}
then the $(\epsilon,\delta)$-LDP constraint is satisfied for all $0\leq \epsilon\leq 1$ \cite{DworkDP,RaviLDP}, which follows as the $\ell_2$-sensitivity of the clipped input $\widetilde{X}$ is $2C$. We also obtain
\begin{align}
    \sigma_{Y}^2 = \sigma_{\widetilde{X}}^2 +\sigma_{\widetilde{Z}}^2
\end{align}
which follows as $\widetilde{X}^n$ and $\widetilde{Z}^n$ are independent.

Next, consider the two corner points of the coordination-randomness rate region, given in (\ref{eq:WCIbound}) and (\ref{eq:MIbound}), for the $(\epsilon,\delta)$-LDP problem defined above. Note that both corner points, i.e., $C(\widetilde{X};Y)$ and $I(\widetilde{X};Y)$, are invariant to mean values of $\widetilde{X}$ and $Y$. Define, for $a\in R$,
\begin{alignat}{2}
    & \{a\}^+= \max\{a,0\},\\
    &\Phi(a) = 0.5 \Big(1+ \erf(\frac{a}{\sqrt{2}})\Big),\\
    & \bar{\beta} = \frac{C}{\sigma_{\widetilde{X}}^2},\\
    & m(a)=\frac{a}{\sigma_Y}.
\end{alignat}

Now, we provide a lower bound on $C(\widetilde{X};Y)$ and a numerical calculation method for $I(\widetilde{X};Y)$.

\begin{theorem}\label{theo:LDPWCIandMIbounds}
    Suppose an RDFC model with a clipped Gaussian input in the range $[-C,C]$ under an $(\epsilon,\delta)$-LDP constraint, imposed via $Q_{\widetilde{X}Y}$, satisfied by using a Gaussian LDP mechanism. We obtain the following lower bound on $C(\widetilde{X};Y)$:
    \begin{align}
        &C(\widetilde{X};Y)\geq \Bigg\{0.5\log\Bigg(1 + \frac{2 \sigma_{\widetilde{X}}}{\sigma_Y - \sigma_{\widetilde{X}}}\Bigg) + \log\Bigg(\frac{\erf\big(\frac{\beta}{\sqrt{2}}\big)}{\sqrt{1-2\gamma(\beta)}}\Bigg) - \gamma(\beta)\Bigg\}^+.\label{eq:LoweronWCI}
    \end{align}
    Furthermore, $I(\widetilde{X};Y)$ can be computed numerically as
    \begin{align}
        I(\widetilde{X};Y) = -\mathbb{E}_{p_Y}[\log(p_Y(Y))] - 0.5\log(2\pi e\sigma_{\widetilde{Z}}^2)\label{eq:MInumeric}
    \end{align}
    where we have the pdf, for $y\in (-\infty,\infty)$,
    \begin{align}
    \displaystyle p_Y(y) = \frac{\phi(m(y))    \left[\displaystyle \Phi \!\left(\frac{\bar{\beta}\sigma_{Y}^2 \!-\! \sigma_{\widetilde{X}} y}{\sigma_{\widetilde{Z}}\sigma_Y}\right) \!-\! \Phi \!\left( - \frac{\bar{\beta}\sigma_{Y}^2 + \sigma_{\widetilde{X}} y}{\sigma_{\widetilde{Z}}\sigma_Y}\right) \right]}{\displaystyle\erf\big(\frac{\bar{\beta}}{\sqrt{2}}\big)\sigma_Y }.\label{eq:pdfofpY}
    \end{align}
\end{theorem}

\begin{proof}[Proof Sketch:]
    The lower bound on the WCI follows from \cite[Theorem~1]{GastparLoweronWCI}, establishing the following lower bound
    \begin{align}
        C(\widetilde{X}; Y) \geq\left\{ C(\widetilde{X}_g; Y_g) + h(\widetilde{X}, Y) - h(\widetilde{X}_g, Y_g)\right\}^+
    \end{align}
    where we have $(\widetilde{X}_g, Y_g) \sim \mathcal{N}(0, K_{\widetilde{X}Y})$ given that the joint pdf $p_{\widetilde{X}, Y}$ satisfies the cross-covariance matrix constraint $K_{\widetilde{X} Y}$. Since $(\widetilde{X}_g, Y_g)$ are jointly Gaussian, we obtain \cite{WCI}
    \begin{align}
        C(\widetilde{X}_g; Y_g) = \frac{1}{2} \log \left( \frac{1 + |\rho_{\widetilde{X}Y}|}{1 - |\rho_{\widetilde{X}Y}|} \right)
    \end{align}
    where we have
    \begin{align}
        \rho_{\widetilde{X}Y} = \frac{\sigma_{\widetilde{X}}}{\sigma_{Y}}=\frac{\sigma_{\widetilde{X}}}{\sqrt{\sigma_{\widetilde{X}}^2+\sigma_{\widetilde{Z}}^2}}
    \end{align}
    as the correlation coefficient between $\widetilde{X}$ and $Y$. Furthermore, we obtain
    \begin{align}
        &h(\widetilde{X}_g, Y_g) = \log(2\pi e\sigma_{\widetilde{X}}\sigma_{\widetilde{Z}}) = \log(2\pi e\sigma_{X}\sqrt{1-2\gamma(\beta)}\sigma_{\widetilde{Z}})
    \end{align}
    which follows by (\ref{eq:sigmaXsquarevssigmaX}). Similarly, we have
    \begin{align}
        &h(\widetilde{X},Y) = h(\widetilde{X}) + h(\widetilde{Z}) = \log\big(2\pi e \sigma_X \erf\big(\frac{\beta}{\sqrt{2}}\big)\sigma_{\widetilde{Z}}\big) - \gamma(\beta)
    \end{align}
    which follows from the properties of the truncated Gaussian random variable $\widetilde{X}$.

    Next, consider the mutual information 
    \begin{align}
        I(\widetilde{X};Y)=h(Y)-h(\widetilde{Z}) = h(Y) - \log(\sqrt{2\pi e}\sigma_{\widetilde{Z}})
    \end{align}
    which can be calculated numerically from the pdf $p_Y$ given in (\ref{eq:pdfofpY}), which is addressed next. We have $Y^n=\widetilde{X}^n+\widetilde{Z}^n$, where $\widetilde{X}^n$ is i.i.d. truncated Gaussian, $\widetilde{Z}^n$ is i.i.d. Gaussian, and $\widetilde{X}^n$ and $\widetilde{Z}^n$ are independent. Therefore, we can obtain the pdf of a sum of independent Gaussian and truncated Gaussian random variables by using \cite[Lemma~3.1]{NTNGauss}.
\end{proof}

{Note that} \eqref{eq:LoweronWCI} can be tighter than the generic inequality in \eqref{eq:MIlessthanWCI}, i.e., the error between the WCI $C(\widetilde{X}; Y)$ and (\ref{eq:LoweronWCI}) can be smaller than the error between the WCI $C(\widetilde{X}; Y)$ and  $I(\widetilde{X};Y)$. A random search sweep identified parameter regimes where \eqref{eq:LoweronWCI} dominates \eqref{eq:MIlessthanWCI}. Table~\ref{table:CompareMIandWCI} lists such examples, with $I(\widetilde{X};Y)$ evaluated numerically via \eqref{eq:MInumeric}.

\begin{table*}[t]
\centering
\captionsetup{
    labelsep=colon, 
    justification=centering,
    font=small, 
}
\bgroup
\def\arraystretch{1.4}
\normalsize
\begin{tabular}{|c|c|c|c|c|c|}
\hline
\rowcolor{blue!15!white}
$\sigma_X$ & $\epsilon$ & $\delta$  & WCI Bound (\ref{eq:LoweronWCI})& $I(\widetilde{X};Y)$ (\ref{eq:MInumeric}) & (\ref{eq:LoweronWCI}) / (\ref{eq:MInumeric})\\
\hline
 0.4938  &   0.8918  &   0.0097  &  0.0324  &   0.0019  & 17.05\\
 \hline
 0.4451  &  0.7917   & 0.0058  &   0.0307  &   0.0012 & 25.58\\ 
 \hline
0.6112   &  0.9256  &   0.0016   &   0.0029   &   0.0019 & 1.53\\
\hline
0.2064   &  0.6399  &   0.0023   &    0.0186  &   0.0003 & 62.00\\
\hline
0.3839   &  0.3637  &   0.0061   &    0.0127  &   0.0002 & 63.50\\
\hline
0.4947   &  0.9884  &   0.0019   &    0.0298  &    0.0018 & 16.56\\
\hline
0.3280   &   0.4663 &    0.0032   &    0.0197 &    0.0002 & 98.50\\ 
\hline
0.3280   &   0.1266 &    0.0032   &   0.0038  &    $1.8\times10^{-5}$ & 214.19\\
\hline
\end{tabular}
\egroup
\vspace{0.0cm}
\caption{WCI lower bound computed from (\ref{eq:LoweronWCI}), numerical computations of $I(\widetilde{X};Y)$ from (\ref{eq:MInumeric}), and their ratios {for random search sweep identified parameters, for which \eqref{eq:LoweronWCI} dominates \eqref{eq:MIlessthanWCI}. These results} illustrate the communication cost gains due to common randomness for the $(\epsilon,\delta)$-LDP constraints imposed.}
\label{table:CompareMIandWCI}
\vspace{-0.1cm}
\end{table*}

Table~\ref{table:CompareMIandWCI} illustrates that common randomness can significantly reduce the communication cost for computing a randomized function output, corresponding to satisfying an $(\epsilon,\delta)$–LDP constraint, by up to \emph{214‑fold} by leveraging semantic communication methods in the RDFC framework. Table~\ref{table:CompareMIandWCI} also reveals a pattern: Better LDP (i.e., smaller $\epsilon$) and smaller input variance (i.e., smaller $\sigma_{X}$) amplify this gain, driving the lower bound on the ratio $C(\widetilde{X};Y)/I(\widetilde{X};Y)$ even higher.

\section{Achievable LDP Parameters Using Non-Asymptotic RDFC Methods}\label{sec:nonasympLDP}
{In this section, we study the finite-blocklength behavior of RDFC under a fixed communication rate. Specifically, we consider any rate $R > I(\tilde{X};Y)$ that lies strictly inside the coordination-randomness region when there is sufficient common randomness. In the RDFC framework, an $(\epsilon,\delta)$-LDP constraint is imposed via the choice of the target joint law $Q_{\widetilde{X}Y}$, as illustrated in Section~\ref{sec:LDP}. However, the encoder-decoder pairs used for RDFC cannot approximate $Q_{\widetilde{X}Y}$ exactly for non-asymptotic regimes, i.e., when $n$ is finite. Our goal in this section is to quantify how fast the induced joint law of a RDFC encoder-decoder pair converges in total variation to the target  law, and hence how fast the effective LDP parameters converge to the desired $(\epsilon,\delta)$ as $n$ grows. Thus, we next establish achievable LDP parameters for a given target joint law synthesizing $(\epsilon,\delta)$-LDP by using a RDFC encoder-decoder pair when there is sufficient common randomness to minimize the communication rate.

Define
\begin{align}
   &\rho^\star := \argmax_{0<\rho\le\frac{1}{2}}
                 \rho(R-I_{\frac{1}{1-\rho}}(\widetilde{X};Y)).
\end{align}

}

\begin{theorem}\label{theo:LDP-parameters}
Let $Q_{\widetilde{X}Y}$ be a target joint law synthesizing $(\epsilon,\delta)$-LDP. Assume that {a given rate satisfies} $R>I(\widetilde{X};Y)>0$ and the moment generating function of the information density $\imath_{\widetilde{X};Y}(\widetilde{x},y)$ is finite in the neighborhood of origin. For every $n\geq 1$, there exists an RDFC encoder-decoder pair such that, for any $a>1$, we obtain
\begin{align}
   &\bigl\| P_{\widetilde{X}^nY^n}{({\widetilde{x}^n,y^n})}-\prod_{i=1}^n Q_{\widetilde{X}Y}{({\widetilde{x}_i,y_i})}\bigr\|_{\mathrm{TV}}\nonumber\\
   &\qquad\le\;a\Delta_n:= aK\begin{cases}
             e^{-\frac{n}{2}\,(R-I_{2}(\widetilde{X};Y))}, &
            \text{ if } \rho^{\star}=\dfrac12,\\[6pt]
            n^{-\frac{1-\rho^\star}{2}}e^{-n\rho^\star\!\big(R-I_{\frac{1}{1-\rho^\star}}(\widetilde{X};Y)\big)}, &\text{ if } \rho^{\star}<\dfrac12
        \end{cases}\label{eq:RDFCYassaeresult}
\end{align}
where $K>0$ is some constant. Consequently, the induced joint law $P_{\widetilde{X}^nY^n}$ synthesizes $(\epsilon,\delta_n)$-LDP with
\begin{align}
\delta_n=\delta+2(e^{\epsilon}+1)a\Delta_n.\label{eq:deltanfromdelta}
\end{align}
\end{theorem}

{
\begin{remark}
Theorem~\ref{theo:LDP-parameters} should be interpreted in a fixed-rate, large-deviations sense. For any rate $R > I(\tilde{X};Y)$ in the interior of the coordination-randomness region, we show that there exists, for each blocklength $n$, a random RDFC encoder-decoder pair of rate $R$ whose induced joint law satisfies the total-variation bound in (\ref{eq:RDFCYassaeresult}). The quantity $\Delta_n$ decays exponentially fast in $n$, and consequently the synthesized mechanism is $(\epsilon,\delta_n)$-LDP with $\delta_n$ given in (\ref{eq:deltanfromdelta}), so that $\delta_n \to \delta$ exponentially. In particular, the rate $R$ itself does not depend on $n$, as our non-asymptotic analysis controls the $n$-dependence of the approximation quality in total variation (and thus of $\delta_n$).
\end{remark}
}

{
Note that the results of Theorem~\ref{theo:LDP-parameters} can be applied to obtain achievable LDP parameters for the RDFC model considered in Section~\ref{sec:LDP}.
}

\begin{proof}[Proof Sketch:]
Consider the likelihood encoder-decoder pair constructed in an i.i.d. manner similar to \cite[Sec.~III]{CuffLikelihoodEncoder}. Using
\cite[Theorem.~3]{YassaeeExtendedSoftCovering} that is an extension of the soft covering lemma for discrete random variables \cite[Lemma~IV.1]{CuffChannelSynthesis}, we have
\begin{equation}
   \mathbb{E}\Big[\bigl\|P_{\widetilde{X}^nY^n}{({\widetilde{x}^n,y^n})}-\prod_{i=1}^n Q_{\widetilde{X}Y}{({\widetilde{x}_i,y_i})}\bigr\|_{\mathrm{TV}}\Big]
   \;\le\;
   \Delta_n.
   \label{eq:explicit-Delta}
\end{equation}
Moreover, by Markov's inequality, the existence of a likelihood encoder-decoder pair that satisfies (\ref{eq:RDFCYassaeresult}) follows for any $a>1$.

Denote the pdf of $\widetilde{X}$ as $p_{\widetilde X}$ and set
\begin{equation}
  m_{\widetilde X}
  :=\operatorname*{ess\,inf}_{\widetilde{x}\in\operatorname{supp}(P_{\widetilde X})}
        p_{\widetilde X}(\widetilde{x})>0 
  \label{eq:def_mX}
\end{equation}
{where $\operatorname*{ess\,inf}$ denotes the essential infimum}. Since total variation distance is non‑expansive under measurable mappings \cite[Eq.~(9)]{TyagiE2209Lecture2}, projecting the pair $(\widetilde X^{n},Y^{n})$ to its first coordinate yields
\begin{equation}
  \bigl\|P_{\widetilde X Y}{({\widetilde{x}^n,y^n})}-Q_{\widetilde X Y}{({\widetilde{x}_i,y_i})}\bigr\|_{\mathrm{TV}}
  \;\le\;
  \bigl\|P_{\widetilde X^{n}Y^{n}}{({\widetilde{x}^n,y^n})}
         -\prod_{i=1}^nQ_{\widetilde X Y}{({\widetilde{x}_i,y_i})}\bigr\|_{\mathrm{TV}}
  \overset{(a)}{\leq} a\Delta_n
  \label{eq:single_letter_TV}
\end{equation}
where $(a)$ follows by (\ref{eq:RDFCYassaeresult}). For any \(\widetilde x\in\operatorname{supp}(P_{\widetilde X})\) and measurable set \(S\subseteq\mathcal Y\), we have
\begin{align}
  \bigl|P_{Y\mid\widetilde X=\widetilde x}(S)
        -Q_{Y\mid\widetilde X=\widetilde x}(S)\bigr|
  &=\frac{1}{p_{\widetilde X}(\widetilde x)}
    \Bigl|
      \int_S\!\!(p_{\widetilde X Y}(\widetilde x,y)
                      -q_{\widetilde X Y}(\widetilde x,y))\,dy
    \Bigr| \nonumber\\
  &\le\frac{2}{m_{\widetilde X}}\,
        \bigl\|P_{\widetilde X Y}-Q_{\widetilde X Y}\bigr\|_{\mathrm{TV}}.
  \label{eq:firstboundonPYgivenXtilde}
\end{align}
Combining \eqref{eq:single_letter_TV} and (\ref{eq:firstboundonPYgivenXtilde}) gives
\begin{equation}
  \bigl|P_{Y|\widetilde X=\widetilde x}(S)
         -Q_{Y|\widetilde X=\widetilde x}(S)\bigr|
  \;\le\;\frac{2a\Delta_n}{m_{\widetilde X}}.
  \label{eq:condTV}
\end{equation}
For any
\(\widetilde x,\widetilde x'\in\operatorname{supp}(P_{\widetilde X})\), for $\widetilde{x}\neq \widetilde{x}'$, and measurable set \(S\subseteq\mathcal Y\), using~\eqref{eq:condTV}
twice and the fact that \(Q_{Y|\widetilde X}\) provides
\((\varepsilon,\delta)\)-LDP, we obtain
\begin{align}
 P_{Y|\widetilde X=\widetilde x}(S)
 &\le Q_{Y|\widetilde X=\widetilde x}(S)+\frac{2a\Delta_n}{m_{\widetilde X}}\\[2pt]
 &\le e^{\varepsilon}Q_{Y|\widetilde X=\widetilde x'}(S)
      +\delta+\frac{2a\Delta_n}{m_{\widetilde X}} \\[2pt]
 &\le e^{\varepsilon}\Bigl[
        P_{Y|\widetilde X=\widetilde x'}(S)
        +\frac{2a\Delta_n}{m_{\widetilde X}}
      \Bigr]+\delta+\frac{2a\Delta_n}{m_{\widetilde X}} \\[2pt]
 &= e^{\varepsilon}P_{Y|\widetilde X=\widetilde x'}(S)
    +\delta
    +\frac{2\bigl(e^{\varepsilon}+1\bigr)a\Delta_n}{m_{\widetilde X}}.
\end{align}
Absorbing \(1/m_{\widetilde X}\) into the constant \(K\) yields \eqref{eq:deltanfromdelta}, so the synthesized
mechanism is \((\varepsilon,\delta_n)\)-LDP.
\end{proof}

{
We note that Theorem~\ref{theo:LDP-parameters} is an existential, non-constructive result, as the bound in (\ref{eq:RDFCYassaeresult}) characterizes the exponential decay of the approximation error but does not correspond to a specific code family. Thus, beyond the fundamental insight that convergence to the target joint law (and to $(\epsilon,\delta)$-LDP) can be exponentially fast, obtaining practical finite-blocklength performance curves would require explicit RDFC code constructions and their numerical evaluation, which we identify as an interesting direction for future work.
}

\section{RDFC for Symmetric Random Response}\label{sec:SymmetricforRR}
\begin{figure}[t]
    \vspace*{-0.0cm}
    \centering
    \includegraphics[width=0.8\textwidth]{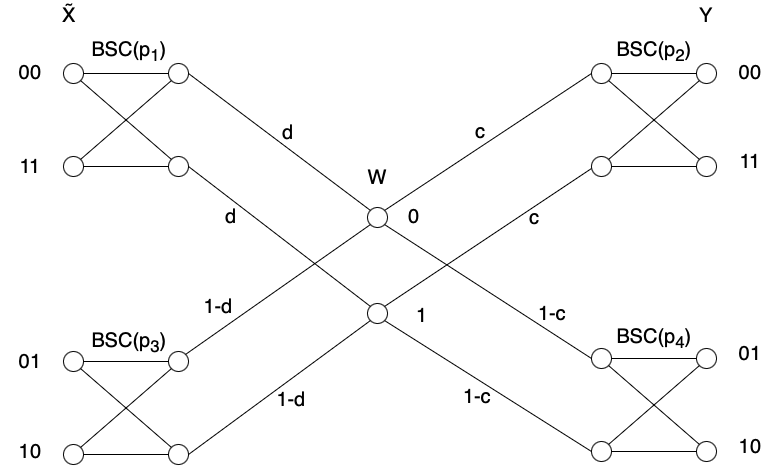}
    \caption{An example random‑response scenario in which the joint distribution $Q_{\widetilde{X}Y}$ can be expressed as a mixture of BSCs. For simplicity, assume that $W$ is uniformly distributed.}\label{fig:symmetric}
    \vspace{-0.0cm}
\end{figure}

We next consider a discrete random‑response scenario, which can be employed as an LDP mechanism \cite{LDPTutorial}, where the transmitter observes $\widetilde{X}^{n}$ and the receiver must generate the random response $Y^{n}$ so that $(\widetilde{X}^{n},Y^{n})$ follows a pmf $Q_{\widetilde{X}Y}^{n}$ that can be represented as in Fig.~\ref{fig:symmetric}. For simplicity, Fig.~\ref{fig:symmetric} shows the 2‑bit‑output case. The binary random variable $W$ in the middle, assumed to be uniformly distributed, drives two channels $P_{\widetilde{X}|W}$ and $P_{Y|W}$, and each of these channels can be written as a mixture of binary symmetric channels (BSCs), making the overall pmf symmetric.

Let $\widetilde{X}=(\widetilde{X}_{1},\widetilde{X}_{2})$ and $Y=(Y_{1},Y_{2})$.  The random response model in Fig.~\ref{fig:symmetric} then admits the following representation
\begin{align}
	\begin{bmatrix}
	\widetilde{X}_1\\
	\widetilde{X}_2\\
	Y_1\\
    Y_2\\
	\end{bmatrix}
	= W 
	\begin{bmatrix}
	1\\
	1\\
	1\\
    1\\
	\end{bmatrix}
	\oplus 
	\begin{bmatrix}
	B_1\\
	B_2\\
	B_3\\
    B_4\\
	\end{bmatrix}\label{eq:corrBSCs}
\end{align}
where $\oplus$ denotes modulo‑2 addition. The noise bits $B_{1}$–$B_{4}$ are mutually dependent but remain independent of $W$.  Because $W$ is uniform, the symmetry property
\begin{align}
    P_{\widetilde{X}_1\widetilde{X}_2Y_1Y_2}(\widetilde{x}_1,\widetilde{x}_2,y_1,y_2) = P_{\widetilde{X}_1\widetilde{X}_2Y_1Y_2}(\bar{\widetilde{x}}_1,\bar{\widetilde{x}}_2,\bar{y}_1,\bar{y}_2)\label{eq:mixtureofBSCSfor3}
\end{align}
holds, where $\bar{y}=1-y$.

Define
\begin{align}
	&q_{Y_1Y_2} = P_{Y_1Y_2|W}(y_1,y_2|0), 
\end{align}
so the subchannel probabilities are
\begin{align}
	&c = q_{00} + q_{11},\\
    &1-c = q_{01} + q_{10},
\end{align}
with the crossover probabilities of the BSCs given as
\begin{align}
    &p_2=q_{11}/c,\\
    &p_4=q_{10}/(1-c).
\end{align}
One can define analogous quantities for $P_{\widetilde{X}|W}$. The same construction extends to joint pmfs that decompose into more than two BSCs \cite{bizimMMMMTIFS}. When a given pmf $Q_{\widetilde{X}Y}$ satisfies the symmetry condition in \eqref{eq:mixtureofBSCSfor3}, computing lower bounds on $C(\widetilde{X};Y)$ becomes considerably simpler, as we illustrate next.

We first state a lower bound on the WCI $C(\widetilde{X};Y)$ by summarizing Theorems~4 and 5 in \cite{WitsenhausenWCI}. Define
\begin{align}
 &\alpha = \sqrt{\frac{kx - 1}{k - 1}},\\
 &x_k^* = \frac{k^2 - 3k + 3}{k(k - 1)},\\
 &f_{1}(x)\! =\! - \frac{2}{k} \big[ (1 + (k - 1)\alpha) \log(1 + (k - 1)\alpha)\nonumber\\
 &\qquad\qquad\qquad+ (k \!-\! 1)(1 \!-\! \alpha) \log(1 \!-\! \alpha) \big] \!+\! 2 \log k,\\
 & f_{2}(x) = 2 \log k - 2(k - 1) \frac{\log(k - 1)}{k - 2} \left(x - \frac{1}{k}\right).
\end{align}
Consider a $k\times k$ joint probability distribution matrix $\textbf{Q}_{\widetilde{X}Y}$ with elements $Q_{i,j}$, representing the pmf $Q_{\widetilde{X}Y}$. For example, we have $k=4$ for the distribution shown in Fig.~\ref{fig:symmetric}. Note that it suffices to consider square matrices $\textbf{Q}_{\widetilde{X}Y}$, as adding or removing rows or columns of zeros does not affect the WCI calculations \cite{WitsenhausenWCI}.

Denote the trace operation as tr$(\cdot)$, and define the maxitrace operation as
\begin{align}
    \maxtr(\textbf{Q}) = \max_{\pi\in\mathcal{S}_k}\sum_{i=1}^k Q_{i,\pi(i)}
\end{align}
where $\mathcal{S}_k$ is the set of all permutations of the indices $\{1,2,\ldots, k\}$.

\begin{theorem}\label{theo:WitsenLower}
    Let $k>2$ and $1/k\leq x \leq 1$, and
    \begin{itemize}
        \item if $x_k^*\leq x\leq 1$, consider 
        \begin{align}
            f(x) \vcentcolon=  f_1(x)
        \end{align}
        \item and, otherwise, if $1/n\leq x\leq x_k^*$, consider
        \begin{align}
            f(x) \vcentcolon=  f_2(x).
        \end{align}
    \end{itemize}
    Given a probability distribution matrix $\textbf{Q}_{\widetilde{X}Y}$, we obtain the following lower bound on $C(\widetilde{X};Y)$
    \begin{align}
        C(\widetilde{X};Y)\geq H(\widetilde{X},Y) - f(\maxtr(\textbf{Q}_{\widetilde{X}Y})).\label{eq:WitsenhausenWCIBound}
    \end{align}
\end{theorem}

Certain symmetric distributions satisfy (\ref{eq:WitsenhausenWCIBound}) with equality \cite{WitsenhausenWCI}.  Restricting attention to joint pmfs $Q_{\widetilde{X}Y}$ that conform to the BSC‑mixture model of Fig.~\ref{fig:symmetric} and (\ref{eq:corrBSCs}), we evaluate the bound given in Theorem~\ref{theo:WitsenLower} and compare it with the mutual information $I(\widetilde{X};Y)$. A random search sweep reveals parameter ranges where (\ref{eq:WitsenhausenWCIBound}) is tighter than the generic inequality in (\ref{eq:MIlessthanWCI}); representative examples are listed in Table~\ref{tab:WitsenhausenResults}.

Table~\ref{tab:WitsenhausenResults} shows that introducing common randomness can remarkably reduce the communication load required to synthesize a randomized function output. Moreover, this gain becomes larger as the mutual information $I(\widetilde{X};Y)$ becomes smaller, i.e., the lower bound on the ratio $C(\widetilde{X};Y)/I(\widetilde{X};Y)$ increases accordingly.  
Even without any common randomness, RDFC methods still offer large communication cost savings, as the WCI lower bound can be as much as $116.55$ times less than the lossless communication rates $H(\widetilde{X})$ and $H(Y)$. Thus, our results illustrate remarkable gains over classical function computation methods by using semantic communication methods for RDFC (both with and without common randomness), which results in {substantial} gains in energy efficiency.

These results motivate further research on, e.g., the minimum achievable communication load for a fixed amount of common randomness, which is important, as common randomness is a true commodity for devices with limited storage.

\begin{table*}[t!]
\centering
\captionsetup{
    labelsep=colon, 
    justification=centering,
    font=small, 
}
\normalsize
\setlength\tabcolsep{1.8pt}      
\renewcommand{\arraystretch}{1.1}
\begin{tabular}{|c|c|c|c|c|c|c|}
\rowcolor{blue!15!white}
\hline
$(p_1, p_2, p_3, p_4, c, d)$ & WCI (\ref{eq:WitsenhausenWCIBound}) & $I(\widetilde{X};Y)$ & $\frac{(\ref{eq:WitsenhausenWCIBound})}{I(\widetilde{X};Y)}$&$H(\widetilde{X})$&$H(Y)$&$\frac{\min\{H(\widetilde{X}),H(Y)\}}{(\ref{eq:WitsenhausenWCIBound})}$\\
\hline
\small $(0.05, 0.45, 0.5, 0.25, 0.45, 0.4)$&0.0977  & 0.0238 & 4.10 & 1.3662  &   1.3813     &    13.984\\
\hline
\small $(0.05, 0.1, 0.2, 0.4, 0.45, 0.3)$  & 0.0998  & 0.0822  &  1.21 & 1.3040  &   1.3813   &    13.069\\
\hline
\small $(0.15, 0.25, 0, 0.05, 0.45, 0.4)$  & 0.3276   & 0.2570  & 1.27 & 1.3662   &  1.3813   &   4.170\\
\hline
\small $(0.25, 0.4, 0.1, 0.25, 0.5, 0.25)$ & 0.0967  & 0.0403 & 2.40 & 1.2555    & 1.3863    &     12.980\\
\hline
\small $(0.45, 0.4, 0.3, 0.1, 0.45, 0.3)$  & 0.1315   & 0.0216 & 6.09 & 1.3040   & 1.3813    &     9.916\\
\hline
\small $(0.3, 0.5, 0.3, 0, 0.5, 0.4)$  & 0.1364   & 0.0411 & 3.32 & 1.3662  &  1.3863 &         10.010\\
\hline
\small $(0.1, 0.05, 0.1, 0, 0.1, 0.05)$ & 0.5325   & 0.3601  & 1.48 & 0.8917  &  1.0182     &   1.675 \\
\hline
\small $(0.5, 0, 0.45, 0.3, 0.4, 0.45)$ & 0.0117  & 0.0014 & 8.36 & 1.3813  &  1.3662 &        116.550 \\
\hline
\end{tabular}
\vspace{0.0cm}
\caption{WCI lower bound computations from (\ref{eq:WitsenhausenWCIBound}), exact computations of $I(\widetilde{X};Y)$ from $Q_{\widetilde{X}Y}$, their ratios, lossless source coding rates $H(\widetilde{X})$ and $H(Y)$ to transmit $\widetilde{X}^n$ and $Y^n$, respectively, and the ratios $\min\{H(\widetilde{X}),H(Y)\}/(\ref{eq:WitsenhausenWCIBound})$. {These values are computed for randomly chosen parameters, for which (\ref{eq:WitsenhausenWCIBound}) provides valid lower bounds.}}
\label{tab:WitsenhausenResults}
\vspace{-0.4cm}
\end{table*}

\section{Conclusions}\label{sec:conclusions}
We introduced the randomized distributed function computation (RDFC) framework as a semantic communication approach that lets a transmitter convey just enough information for the receiver to synthesize a (privacy‑preserving) randomized function of the input data. By framing privacy mechanisms, such as Gaussian local differential privacy and symmetric random response, within a strong coordination model, we demonstrated that RDFC can reduce the required communication rate by up to two orders of magnitude compared to lossless compression methods. Note that such significant gains in the communication rate, both with and without common randomness, result in {substantial} gains in energy efficiency. We further showed, through new bounds and numerical examples, how sufficient amount of common randomness translated into ultra-efficient semantic communication rates, while a finite blocklength analysis confirmed that the resulting privacy guarantees tightened exponentially fast with the sequence length. These findings highlighted RDFC as a practical route to privacy-aware and energy-efficient communication in future distributed computation systems. {As an interesting direction for future work, recent variational formulations of the WCI \cite{WCIVariational} can be adapted to the RDFC setting to provide variational upper bounds on the WCI, complementing the lower bounds developed in this paper.}

\backmatter

\section*{Declarations}

\subsection*{Availability of data and materials}
Data sharing is not applicable to this article as no datasets were generated or analysed during the current study.
\subsection*{Competing interests}
The author declares that he has no competing interests.
\subsection*{Funding}
This work was supported in part by the ZENITH Research and Leadership Career Development Fund, Swedish Foundation for Strategic Research (SSF), and EU COST Action 6G-PHYSEC.
\subsection*{Author's contributions}
OG is the sole author who prepared, read, and approved the final manuscript.
\subsection*{Acknowledgements}
OG thanks anonymous reviewers of the conference version \cite{MyWIFS2024} of this work for their insightful comments, including the suggestion of varying only one parameter in Table~\ref{table:CompareMIandWCI} that led to the remarkable communication load gain result of $214$ times, as given in the last row of Table~\ref{table:CompareMIandWCI}.
\subsection*{Author's information}
OG received the B.Sc. degree (Highest Distinction) in Electrical and Electronics Engineering from Bilkent University, Turkey in 2011; M.Sc. (Highest Distinction) and Dr.-Ing. (Ph.D. equivalent) degrees in Communications Engineering both from the Technical University of Munich (TUM), Germany in 2013 and 2018, respectively. He was a Working Student in the Communication Systems division of Intel Mobile Communications (IMC), now Apple Inc., in Munich, Germany during November 2012 - March 2013. Onur worked as a Research and Teaching Assistant at TUM Chair of Communications Engineering (LNT) between February 2014 - May 2019. As a Visiting Researcher, among more than twenty Research Stays at Top Universities and Companies, he was at TU Eindhoven, Netherlands during February 2018 - March 2018. Onur was a Visiting Research Group Leader at Georgia Institute of Technology, Atlanta, USA during February 2022 - March 2022. He was also a Visiting Professor at TU Dresden, Germany during February 2023 - March 2023. Following Research Associate and Group Leader positions at TUM, TU Berlin, and the University of Siegen, he joined Linköping University in October 2022 as an ELLIIT Assistant Professor and obtained tenure as an Associate Professor leading the Information Theory and Security Laboratory (ITSL) in August 2024. He became a Swedish Docent (Dr.-habil.) of Information Theory in December 2023 and an IEEE Senior Member in July 2024. Since September 2025, Onur has been a Tenured Full Professor leading the Lehrstuhl für Nachrichtentechnik (Institute of Communications Engineering ) at TU Dortmund (TUDO), Germany and a Guest Professor at Linköping University, Sweden.

He has received the 2025 IEEE Information Theory Society – Joy Thomas Tutorial Paper Award, the 2023 ZENITH Research and Career Development Award, 2021 IEEE Transactions on Communications - Exemplary Reviewer Award, and the prestigious VDE Information Technology Society (ITG) 2021 Johann-Philipp-Reis Award. His research interests include distributed function computation, information-theoretic privacy and security, coding theory, integrated sensing and communication, and private learning. Among his publications is the book \emph{Key Agreement with Physical Unclonable Functions and Biometric Identifiers} (Dr. Hut Verlag, 2019). 

He serves as Associate Editor for \textsc{IEEE JOURNAL ON SELECTED AREAS IN COMMUNICATIONS}, \textsc{IEEE TRANSACTIONS ON COMMUNICATIONS}, and \textsc{ENTROPY} Journal, and was recently an Associate Editor of \textsc{EURASIP JOURNAL ON WIRELESS COMMUNICATIONS AND NETWORKING} and a Guest Editor of \textsc{IEEE JOURNAL ON SELECTED AREAS IN INFORMATION THEORY}. He also serves as a Board Member and Secretary of the IEEE Sweden VT/COM/IT Joint Chapter, and as a Working Group Leader for EU COST Action 6G-PHYSEC.

\bibliography{sn-bibliography}

\end{document}